\documentclass[cmp,envcountsect,envcountsame,envcountreset]{svjour_arxiv}

\usepackage{amsfonts,amssymb,graphicx}
\usepackage{epstopdf}
\usepackage{cite}
\usepackage{indentfirst}
\usepackage{newlfont}
\usepackage{latexsym}
\usepackage{amsmath, amssymb, mathrsfs, dsfont, mathtools, empheq,rotating}
\usepackage{graphicx}
\usepackage{comment}
\usepackage{color}

\numberwithin{equation}{section}

\def\be{\begin{equation}}
\def\ee{\end{equation}}
\newcommand{\eq}{\begin{equation}}
\newcommand{\eeq}{\end{equation}}
\DeclareMathOperator*{\argmax}{argmax}
\DeclareMathOperator{\dist}{dist}

\newcommand{\virg}[1]{``#1''}

\newcommand{\R}{\mathbb R}
\newcommand{\N}{\mathbb N}

\newcommand{\Z}{\mathbb Z}
\newcommand{\C}{\mathbb C}
\newcommand{\HH}{\mathbb H}
\newcommand{\E}{\mathbb E}

\newcommand{\Prob}{\mathbb P}

\newcommand{\1}{\mathds 1}
\renewcommand{\i}{\mathrm i}
\newcommand{\dd}{\mathrm d}
\renewcommand{\d}{\mathrm d}
\newcommand{\eps}{\varepsilon}

\renewcommand{\rho}{\varrho}

\newcommand{\DD}{\mathscr D}
\newcommand{\xv}{\boldsymbol{x}}

\newcommand{\Exv}{\E_{\xv}}

\newcommand{\Ex}{\E_x}

\newcommand{\Exi}{\E_\xi}

\newcommand{\Him}{H} 
\newcommand{\Zim}{Z} 
\newcommand{\muim}{\mu} 
\newcommand{\pim}{p} 
\newcommand{\mim}{m} 
\newcommand{\vpim}{\widetilde{p}\,}

\newcommand{\Zcl}{Z^{\textup{(0)}}}

\newcommand{\hx}{\hat{x}}
\newcommand{\xn}{\hx_n}
\newcommand{\poo}{p^{(0)}}

\newcommand{\ZclN}{Z^{(0)}_{N}}
\newcommand{\pclN}{p^{(0)}_{N}}
\newcommand{\muclN}{\mu^{(0)}_{N}}

\newcommand{\bmat}{\begin{pmatrix}}
\newcommand{\emat}{\end{pmatrix}}

\journalname{}

\begin{document}

\title{Mean-Field Monomer-Dimer models. A review.}

\author{Diego Alberici \and Pierluigi Contucci \and Emanuele Mingione }

\institute{
			University of Bologna - Department of Mathematics \\
            Piazza di Porta San Donato 5, Bologna (Italy)\\
            \email{diego.alberici2@unibo.it, pierluigi.contucci@unibo.it, emanuele.mingione2@unibo.it}
            }

\date{}


\maketitle
\vskip .5truecm
\rightline{\it To Chuck Newman, on his 70th birthday}
\vskip 1truecm
\begin{abstract}
A collection of rigorous results for a class of mean-field monomer-dimer models is presented.
It includes a Gaussian representation for the partition function that is shown to considerably simplify
the proofs. The solutions of the quenched diluted case and the random monomer case are explained.
The presence of the attractive component of the Van der Waals potential is considered and the coexistence phase
coexistence transition analysed. In particular the breakdown of the central limit theorem
is illustrated at the critical point where a non Gaussian, quartic exponential distribution is found
for the number of monomers centered and rescaled with the volume to the power $3/4$.
\end{abstract}


\section{Introduction}
The monomer-dimer models, an instance in the wide set of interacting particle systems,
have a relevant role in equilibrium statistical mechanics. They were introduced
to describe, in a simplified yet effective way, the process of absorption of
monoatomic or diatomic molecules in condensed-matter
physics \cite{Rob,Ch,Chcambr} or the behaviour of liquid solutions composed by molecules of different sizes \cite{FR}.
Exact solutions in specific cases (e.g. the perfect matching problem) have been obtained on planar lattices \cite{F,K,TFK,Ltra,GJL} and the problem on regular lattices is also interesting for the liquid crystals modelling \cite{O,HLliquid,DG,Alb,GL}.
The impact and the interest that monomer-dimer models have attracted has progressively grown
beyond physics. Their thermodynamic
behaviour has indeed proved to be useful in computer science for the matching
problem \cite{KS, BN} or for the applications of statistical physics methods to the social
sciences \cite{BCSV}.

From the physical point of view monomers and
dimers cannot occupy the same site of a lattice due to the strong repulsion
generated by the Pauli exclusion principle. Beside this interaction though, as already
noticed by Peierls \cite{Pei} in the first theoretical
physics accounts, the attractive component of the Van der Waals potentials might
influence the phase structure of the model and its thermodynamic behaviour.
With the contemporary presence of those two interactions the
global physical observables become particularly difficult to study. Generic Gaussian
fluctuations on each ergodic component can still be expected but the nature of the
critical point, if any, is a priori not obvious.

Here we focus on a set of monomer-dimer models in the mean field setting,
i.e. on the complete graph where every site interacts in average with any
other, and present a review of recent results. Section 2 introduces the general properties of the monomer-dimer
systems that we approach with the help of a Gaussian representation for their
partition function. This representation and its combinatorial features help to embed
and ease part of the classical difficulties of their studies. The celebrated Heilmann and Lieb relation, so rich
of rigorous consequences, emerges as the formula of integration by parts for Gaussian random vectors.
The absence of phase transition for the pure hard-core case is therefore
derived. Section 3 treats two quenched cases, namely the diluted
complete graph of Erd\H{o}s-R\'enyi type as well as the diluted random
monomer field activity. For both cases we compute the exact solution. The diluted
graph is treated with the help of correlation inequalities and the representation
of the monomer density as the solution of an iterative distributional equation.
The random monomer activity model is solved by reducing the computation
of the equilibrium state to a standard variational problem, again with the help of
the Gaussian representation. Section 4 introduces a genuine deterministic
mean field model with and without the attractive interaction. It is shown how
the model with attraction displays a phase space structure similar to the
mean field ferromagnet but without the usual plus-minus symmetry. The model
has a coexistence line bounded by a critical point with standard mean-field critical exponents.
In Section 5 we show that while outside
the critical point the central limit theorem holds, at criticality it breaks down and
the limiting distribution is found at a scale of $N^{3/4}$ and turns out to be
a quartic exponential, like in the well known results by Newman and Ellis \cite{ellis1978statistics,ellis1978limit}
for the ferromagnet.


\section{Definition and general properties} \label{sec: definitions}
Let $G=(V,E)$ be a finite undirected graph with vertex set $V$ and edge set $E\subseteq\{ij\equiv\{i,j\}\,|\, i\in V,\,j\in V,\,i\neq j\}\,$.

\begin{definition}[Monomer-dimer configurations]
A set of edges $D\subseteq E$ is called a \textit{monomer-dimer configuration}, or a \textit{matching}, if the edges in $D$ are pairwise non-incident.
The space of all possible monomer-dimer configurations on the graph $G$ is denoted by $\DD_G$.
\end{definition}

Given a monomer-dimer configuration $D$, we say that every edge in $D$ is occupied by a \textit{dimer}, while every vertex that does not appear in $D$ is occupied by a \textit{monomer}.
The set of monomers associated to $D$ is denoted by $M(D)$.

\begin{figure}[h]
\centering
\includegraphics[scale=0.2]{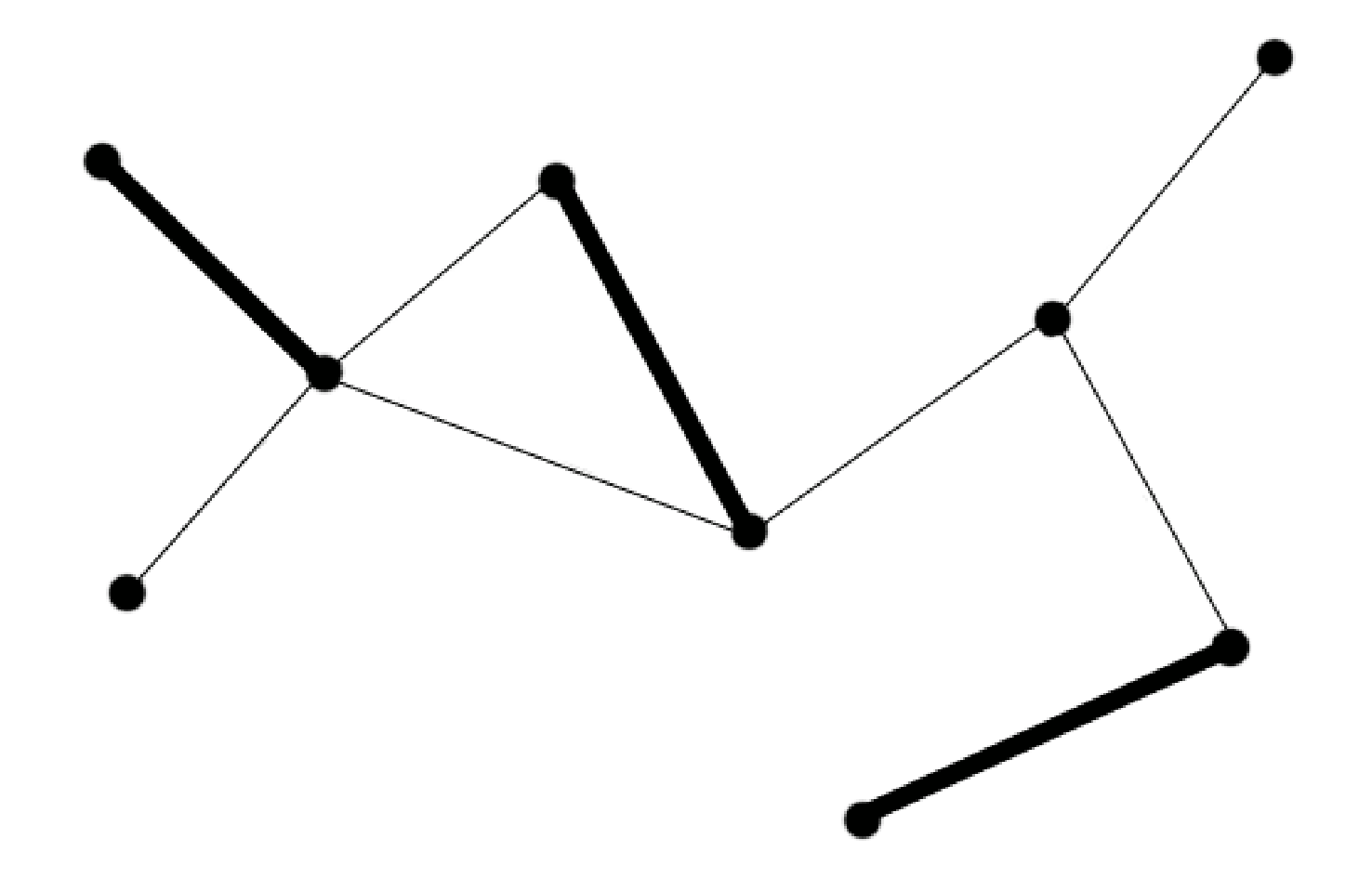} \quad
\includegraphics[scale=0.2]{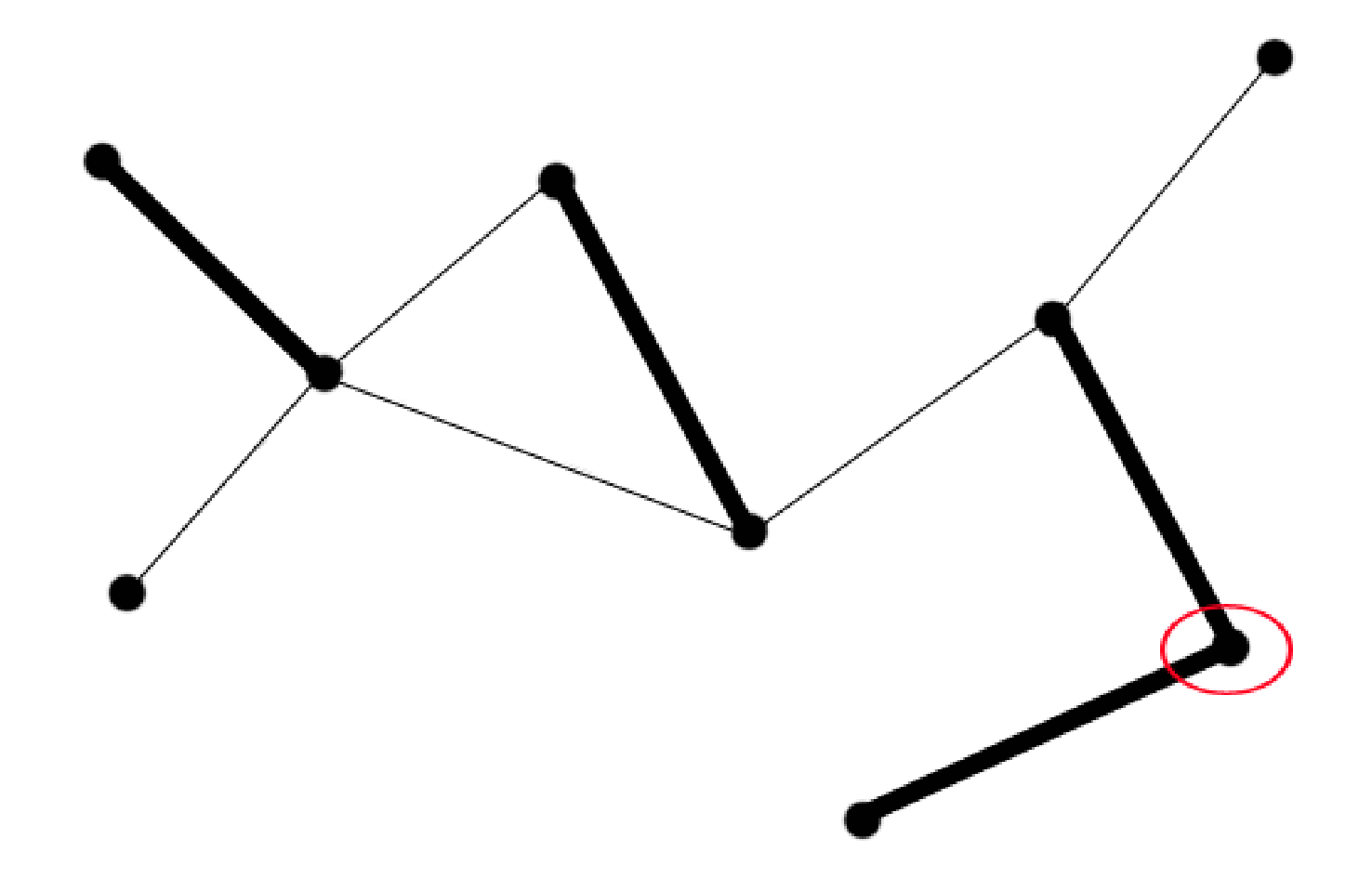}
\caption{The bold edges in the left figure form a monomer-dimer configuration on the graph, while those in the right figure do not because two of them share a vertex.}
\end{figure}

\begin{remark} \label{rk0: alpha}
We can associate the dimer occupation variable $\alpha_{ij}\in\{0,1\}$ to each edge $ij\in E\,$: the edge $ij$ is occupied by a dimer if and only if $\alpha_{ij}$ takes the value $1$.
It is clear that monomer-dimer configurations are in one-to-one correspondence with vectors $\alpha\in\{0,1\}^E$ satisfying the following constraint:
\eq \label{eq0: hard-core}
 \sum_{j\sim i}\,\alpha_{ij} \,\leq\, 1 \,,\quad\forall\,i\in V
\eeq
where $j\sim i$ means that $ij\in E$.
Therefore, with a slight abuse of notation, we denote by $\DD_G$ also the set of $\alpha\in\{0,1\}^E$ that satisfy eq. \eqref{eq0: hard-core}.
The condition \eqref{eq0: hard-core} guarantees that at most one dimer can be incident to a given vertex $i$, namely two dimers cannot be incident. This fact is usually referred as \textit{hard-core interaction} or \textit{hard-core constraint} or \textit{monogamy constraint}.
We also introduce an auxiliary variable, the monomer occupation variable,
\eq \label{eq0: alphai}
\alpha_i:=1-\sum_{j\sim i}\alpha_{ij}\in\{0,1\}
\eeq
for each vertex $i\in V$: the vertex $i$ is occupied by a monomer if and only if $\alpha_i$ takes the value $1$.
\end{remark}

The definition of monomer-dimer configurations already allows to raise non-trivial combinatorial questions as ``How many monomer-dimer configurations, for a fixed number of dimers, exist on given a graph $G$?''.
This combinatorial problem is known to be NP-hard in general, but there are polynomial algorithms and exact solutions for specific cases \cite{K, F, TFK, KS, HL}.
In Statistical Mechanics a further structure is introduced and the previous problem becomes a specific limit case. We consider a Gibbs probability measure on the set of monomer-dimer configurations. There are several choices for the measure, depending on how we decide to model the interactions in the system.

\subsection{Pure hard-core interaction} \label{sec: hard-core}
This amounts to take into account only the hard-core interaction among particles and assign a dimer activity $w_{ij}\geq0$ to each edge $ij\in E$ and a monomer activity $x_i>0$ to each vertex $i\in V$.

\begin{definition}[Monomer-dimer models with pure hard-core interaction] \label{df0: md model}
A \textit{ pure monomer-dimer model} on $G$ is given by the following probability measure on $\DD_G$:
\eq \label{eq0: measure}
\mu_{G}(D) \,:=\, \frac{1}{Z_G}\; \prod_{ij\in D}w_{ij} \prod_{i\in M(D)}\!\!x_i \quad\forall\,D\in\DD_G \;,
\eeq
where the normalizing factor, called partition function, is
\eq \label{eq0: pf}
Z_G := \sum_{D\in\DD_G} \prod_{ij\in D}w_{ij} \prod_{i\in M(D)}\!\!x_i \;.
\eeq
We denote by $\langle\,\cdot\,\rangle_G$ the expectation w.r.t. the measure $\mu_G$.
The dependence of the measure on the activities $w_{ij},\,x_i$ is usually implicit in the notations.
\end{definition}

\begin{remark} \label{rk0: hamiltonian}
Equivalently, one can think  the measure \eqref{eq0: measure} as a Gibbs measure on the space of  occupancy variables $\alpha$ (see Remark \ref{rk0: alpha}), namely
\eq\label{eq0: measure-hami}
\mu_G(\alpha)\,=\, \frac{1}{Z_G}\,e^{-H_G(\alpha)} \quad\forall\,\alpha\in\DD_G \;,
\eeq
where
\eq \label{eq0: hamilt}
H_G(\alpha) \,:=\, -\sum_{ij\in E}h_{ij}\alpha_{ij} \,-\, \sum_{i\in V}h_i\alpha_i \quad\forall\,\alpha\in\DD_G\;
\eeq
is the \textit{Hamiltonian function} and  $h_{i}:=\log x_i$, $h_{ij}:=\log w_{ij}$ are called \textit{monomer, dimer field} respectively. The partition function \eqref{eq0: pf} rewrites
\eq \label{eq0: pf-hami}
Z_G = \sum_{\alpha\in\DD_G} \exp\big(\,\sum_{ij\in E}h_{ij}\alpha_{ij} + \sum_{i\in V}h_i\alpha_i\big)\;.
\eeq
\end{remark}

\begin{remark} \label{rk0: riduz-coeff}
It is worth to notice that the definition \ref{df0: md model} is redundant for two reasons.
First one can consider without loss of generality monomer-dimer models on complete graphs only: a monomer-dimer model on the graph $G=(V,E)$ coincides with a monomer-dimer model on the complete graph with $N=|V|$ vertices, by taking $w_{ij}=0$ for all pairs $ij\notin E$. In this case we denote the partition function \eqref{eq0: pf-hami} with $Z_N$. Secondly, one can set without loss of generality all the monomer activities equal to $1$: the monomer-dimer model with activities $(w_{ij},x_i)$ coincides with the monomer-dimer model with activities $(\frac{w_{ij}}{x_ix_j},1)$, since the relation
\[ \prod_{i\in M(D)}\!\!x_i \,=\, \bigg(\prod_{i\in V}x_i\bigg) \prod_{ij\in D}\frac{1}{x_i\,x_j} \;.\]
shows that the partition function is multiplied by an overall constant and therefore
the probability measure is left unchanged. The same argument shows also that if the dimer activity is uniform on the graph then it can be set equal to $1$: the monomer-dimer model with activities $(w,x_i)$ coincides with the monomer-dimer model with activities $(1,\frac{x_i}{\sqrt w})$, since
\[ w^{|D|} \,=\, w^{N/2}\, \bigg(\frac{1}{\sqrt w}\bigg)^{|M(D)|} \;.\]
\end{remark}

\begin{remark} \label{rk0: pressure bounds}
The following bounds for the pressure (logarithm of the partition function) will be useful:
\eq \label{eq0: p bounds}
\sum_{i\in V}\log x_i \,\leq\, \log Z_{G} \,\leq\,
\sum_{i\in V}\log x_i + \sum_{ij\in E}\log\big(1+\frac{w_{ij}}{x_i\,x_j}\big) \;.
\eeq
The lower bound is obtained considering only the configuration with no dimers, 
while the upper bound is obtained by eliminating the hard-core constraint.
\end{remark}

An interesting fact about monomer-dimer models is that they are strictly related to Gaussian random vectors.

\begin{proposition}[Gaussian representation \cite{ACMrand,V}] \label{prop: gauss repr}
The partition function of any monomer-dimer model over $N$ vertices can be written as
\eq \label{eq: gauss repr}
Z_N \,=\, \Exi\bigg[\prod_{i=1}^N(\xi_i+x_i)\bigg] \;,
\eeq
where $\xi=(\xi_1,\dots,\xi_N)$ is a Gaussian random vector with mean $0$ and covariance matrix $W=(w_{ij})_{i,j=1,\dots,N}$ and $\Exi[\,\cdot\,]$ denotes the expectation with respect to $\xi$.
The diagonal entries $w_{ii}$ are arbitrary numbers, chosen in such a way that $W$ is a positive semi-definite matrix.
\end{proposition}

\proof
The monomer-dimer configurations on the complete graph are all the partitions into pairs of any set $A\subseteq\{1,\dots,N\}$, hence
\eq \label{eq: gauss repr proof1}
Z_N \,= \sum_{D\in\DD_N}\, \prod_{ij\in D}w_{ij} \prod_{\,i\in M(D)}\!\!x_i \,=
\sum_{A\subseteq\{1,\dots,N\}} \sum_{P\text{ partition}\atop\text{ of }A\text{ into pairs}}\, \prod_{ij\in P}w_{ij}\, \prod_{i\in A^c}x_i \;.
\eeq
Now choose $w_{ii}$ for $i=1,\dots,N$ such that the matrix $W=(w_{ij})_{i,j=1,\dots,N}$ is positive semi-definite\footnote{For example one can choose $w_{ii}\geq\sum_{j\neq i}w_{ij}$ for every $i=1,\dots,N$.}.
Then there exists an (eventually degenerate) Gaussian vector $\xi=(\xi_1,\dots,\xi_N)$ with mean $0$ and covariance matrix $W$. And by the Isserlis-Wick rule
\eq \label{eq: gauss repr proof2}
\Exi\bigg[\prod_{i\in A}\xi_i\bigg] \;= \sum_{P\text{ partition}\atop\text{ of }A\text{ into pairs}} \prod_{ij\in P}w_{ij} \;.
\eeq
Substituting \eqref{eq: gauss repr proof2} into \eqref{eq: gauss repr proof1} one obtains \eqref{eq: gauss repr}. \hfill \qed
\endproof

We notice that  the representation \eqref{eq: gauss repr} allows to express average values  w.r.t. the measure \eqref{eq0: measure} as Gaussian averages. For example, given a vertex $i\in V$, its monomer probability by \eqref{eq0: pf-hami} writes
\eq \label{eq: monomers average}
\left\langle \alpha_i \right\rangle_N\, = \,\frac{\partial}{\partial h_i}\log Z_N
\eeq
then, using the representation \eqref{eq: gauss repr} in r.h.s. of \eqref{eq: monomers average} together with the identity
$\frac{\partial}{\partial h_i}\equiv x_i\frac{\partial}{\partial x_i}$, one obtains
\eq \label{eq: gauss repr monomers average}
\left\langle \alpha_i \right\rangle_N \, = \,x_i\, \,\Exi\bigg[\dfrac{1}{\xi_i+x_i} g_N(\xi,x)\bigg]
\eeq
where $g_N(\xi,x)=\dfrac{1}{Z_N}\,\prod_{i=1}^N(\xi_i+x_i)$.\\

Heilmann and Lieb \cite{HL} provided a recursion for the partition functions of monomer-dimer models. As we will see this is a fundamental tool to obtain exact solutions and to prove general properties.

\begin{proposition}[Heilmann-Lieb recursion \cite{HL}] \label{prop: HL rec}
Fixing any vertex $i\in V$ it holds:
\eq \label{eq: HL rec}
Z_G \,=\, x_i\,Z_{G-i} \,+\, \sum_{j\sim i}\,w_{ij}\,Z_{G-i-j} \;.
\eeq
Here $G-i$ denotes the graph obtained from $G$ deleting the vertex $i$ and all its incident edges.
\end{proposition}

The Heilmann-Lieb recursion can be obtained directly from the definition \eqref{eq0: pf}, exploiting the hard-core constraint: the first term on the r.h.s. of \eqref{eq: HL rec} corresponds to a monomer on $i$, while the following terms correspond to a dimer on $ij$ for some $j$ neighbour of $i$.
Here we show a different proof that uses Gaussian integration by parts.

\proof[see \cite{ACMrand}]
Set $N:=|V|$. Introduce zero dimer weights $w_{hk}=0$ for all the pairs $hk\notin E$, so that $Z_G\equiv Z_N$.
Following the proposition \ref{prop: gauss repr}, introduce an $N$-dimensional Gaussian vector $\xi$ with mean $0$ and covariance matrix $W$. Then write the identity \eqref{eq: gauss repr} isolating the vertex $i\,$:
\eq \label{eq: HL rec proof1}
Z_G \,=\, \Exi\bigg[\prod_{k=1}^N(\xi_k+x_k)\bigg] \,=\,
x_i\;\Exi\bigg[\prod_{k\neq i}(\xi_k+x_k)\bigg] \,+\, \Exi\bigg[\xi_i\,\prod_{k\neq i}(\xi_k+x_k)\bigg] \;.
\eeq
Now apply the Gaussian integration by parts to the second term on the r.h.s. of \eqref{eq: HL rec proof1}:
\eq \label{eq: HL rec proof2}
\Exi\bigg[\xi_i\,\prod_{k\neq i}(\xi_k+x_k)\bigg] \,=\,
\sum_{j=1}^N\, \Exi[\xi_i\xi_j]\; \Exi\bigg[ \frac{\partial}{\partial \xi_j}\prod_{k\neq i}(\xi_k+x_k) \bigg] \,=\,
\sum_{j\neq i}\, w_{ij}\; \Exi\bigg[\prod_{k\neq i,j}(\xi_k+x_k) \bigg] \;.
\eeq
Notice that summing over $j\neq i$ in the r.h.s. of \eqref{eq: HL rec proof2} is equivalent to sum over $j\sim i$, since by definition $w_{ij}=0$ if $ij\notin E$. Substitute \eqref{eq: HL rec proof2} into \eqref{eq: HL rec proof1}:
\eq \label{eq: HL rec proof3}
Z_G \,=\,
x_i\;\Exi\bigg[\prod_{k\neq i}(\xi_k+x_k)\bigg] \,+\, \sum_{j\sim i}\, w_{ij}\; \Exi\bigg[\prod_{k\neq i,j}(\xi_k+x_k) \bigg] \;.
\eeq
To conclude the proof observe that $(\xi_k)_{k\neq i}$ is an $(N-1)$-dimensional Gaussian vector with mean $0$ and covariance $(w_{hk})_{h,k\neq i}$. Hence by proposition \ref{prop: gauss repr}
\eq \label{eq: HL rec proof4}
Z_{G-i} \,=\, \Exi\bigg[\prod_{k\neq i}(\xi_k+x_k)\bigg] \;.
\eeq
And similarly
\eq \label{eq: HL rec proof5}
Z_{G-i-j} \,=\, \Exi\bigg[\prod_{k\neq i,j}(\xi_k+x_k)\bigg] \;.
\eeq
\endproof

The main general result about monomer-dimer models is the absence of phase transitions, proved by Heilmann and Lieb \cite{HL,HLprl}.
This result is obtained by localizing the complex zeros of the partition functions far from the positive real axes.
A different probabilistic approach has been later proposed by van den Berg \cite{vdB}.

\begin{theorem}[Zeros of the partition function \cite{HL}] \label{thm0: zeros}
Consider uniform monomer activity $x$ on the graph and arbitrary dimer activities $w_{ij}$.
The partition function $Z_G$ is a polynomial of degree $N$ in $x$, where $N=|V|$.
The complex zeros of $Z_G$ are purely imaginary:
\eq \label{eq0: zeros}
\{x\in\C \,|\, Z_G(w_{ij},x)=0\} \,\subset\, \i\,\R \;.
\eeq
Furthermore they interlace the zeros of $Z_{G-i}$ for any given $i\in V$, that is:
\eq \label{eq0: zeros-inter}
a_1 \leq a_1' \leq a_2 \leq a_2' \leq \dots \leq a_{N-1}' \leq a_N \;,
\eeq
where $-\i a_1,\dots,-\i a_N$ are the zeros of $Z_G$ and $-\i a_1',\dots,-\i a_{N-1}'$ are the zeros of $Z_{G-i}$.
The relation \eqref{eq0: zeros-inter} holds with strict inequalities if $w_{ij}>0$ for all $i,j\in V$.
\end{theorem}

\begin{corollary}[Absence of phase transitions]
Consider dimer activities $w\,w_{ij}^{(N)}$ and monomer activities $x\,x_i^{(N)}$ and assume they are chosen in such a way that $p:=\lim_{N\to\infty}\frac{1}{N}\log Z_N$ exists.
Then the function $p$ is analytic in the variables $(w,x)\in(0,\infty)^2$ and the derivatives $\frac{\partial^{h+k}}{\partial^h w\,\partial^k x}$ can be interchanged with the limit $N\to\infty$.
\end{corollary}

\subsection{Hard-core and imitative interactions} \label{sec: imit}

Beyond the hard-core constraint it is possible to enrich monomer-dimer models with other kinds of interaction.
For example in this work we consider, for a given $D\in\DD_G$, the set of edges connecting particles of the same kind
\eq\label{eq0: imitation sets}
I(D)=\{ij\in E \,|\, i,j\in M(D)\text{ or } i,j\notin M(D)\}
\eeq
and we introduce an interaction between any pair of vertices $ij\in I(D)$ tuned by a coupling $J_{ij}\in\R$. More precisely

\begin{definition}[Monomer-dimer models with imitative interactions] \label{df0: md model-im}
An \textit{imitative monomer-dimer model} on $G$ is given by the following Gibbs probability measure on $\DD_G$:
\eq \label{eq0: measure-im}
\mu_{G}(D) \,:=\, \frac{1}{Z_G}\; \prod_{ij\in D}w_{ij} \prod_{i\in M(D)}\!\!x_i \prod_{ij\in I(D)}\!\!e^{J_{ij}}
\eeq
for all $D\in\DD_G$.
The partition function is
\eq \label{eq0: measure-im}
Z_G \,:=\, \sum_{D\in\DD_G} \prod_{ij\in D}w_{ij} \prod_{i\in M(D)}\!\!x_i \prod_{ij\in I(D)}\!\!e^{J_{ij}}
\eeq
The dependence of the measure on the coefficients $w_{ij},\,x_i,\,J_{ij}$ is usually implicit in the notations.
\end{definition}

When all the $J_{ij}$'s take the value zero this model is the pure hard-core model introduced in the previous section.
Positive values of the $J_{ij}$'s favour the configurations with clusters of dimers and clusters of monomers.

\begin{remark} \label{rk0: hami-imit}
The usual Gibbs form $\frac{1}{Z_G}\,e^{-H_G(\alpha)}$ for the measure \eqref{eq0: measure-im} is obtained by setting $x_i=:e^{h_i}$, $w_{ij}=:e^{h_{ij}}$ and taking as Hamiltonian function
\eq \label{eq0: hamilt-im}
H_G(\alpha) \,:=\, -\sum_{ij\in E}h_{ij}\alpha_{ij} - \sum_{i\in V}h_i\alpha_i - \sum_{ij\in E} J_{ij} \big( \alpha_i\alpha_j + (1-\alpha_i)(1-\alpha_j) \big)
\eeq
for all $\alpha\in\DD_G$.
\end{remark}

The Gaussian representation and the recursion relation found for the pure hard-core case can be extended
to the imitative case.

\begin{proposition} \label{prop: gauss-repr imit}
The partition function of any monomer-dimer model over $N$ vertices can be written as
\eq \label{eq: gauss-repr imit}
Z_N \,=\, \Exi\Big[\sum_{A\subset\{1,\dots,N\}}\, \prod_{i\in A}\xi_i\; \prod_{i\in A^c}x_i\; \prod_{\substack{i,j\in A\text{ or}\\i,j\in A^c}}e^{J_{ij}/2} \Big] \;,
\eeq
where $\xi=(\xi_1,\dots,\xi_N)$ is a Gaussian random vector with mean $0$ and covariance matrix $W=(w_{ij})_{i,j=1,\dots,N}$ and $\Exi[\,\cdot\,]$ denotes the expectation with respect to $\xi$.
The diagonal entries $w_{ii}$ are arbitrary numbers, chosen in such a way that $W$ is a positive semi-definite matrix.
Moreover we set $J_{ii}=0\,$.
\end{proposition}

The proof is the same as proposition \ref{prop: gauss repr}.
It is interesting to observe that, when all the $\xi_i$'s are positive, the sum inside the expectation on the r.h.s. of \eqref{eq: gauss-repr imit} is the partition function of an Ising model.

\begin{proposition} \label{prop: HL-rec imit}
Fixing any vertex $i\in V$ it holds:
\eq \label{eq: HL-rec imit}
Z_G \,=\, x_i\,\widetilde Z_{G-i} \,+\, \sum_{j\sim i}\,w_{ij}\,\widetilde Z_{G-i-j} \;,
\eeq
where:
\begin{itemize}
\item in the partition function $\widetilde Z_{G-i}$ the monomer activity $x_k$ is replaced by $x_k\,e^{J_{ik}}$ for every vertex $k$ (notice that only the neighbours of $i$ actually change their activities);
\item in the partition function $\widetilde Z_{G-i-j}$ the dimer activity $w_{kk'}$ is replaced by $w_{kk'}\,e^{J_{ik}+J_{ik'}+J_{jk}+J_{jk'}}$ for all vertices $k,k'$ (notice that only the neighbours of $i$ or $j$ actually change their activities).
\end{itemize}
\end{proposition}

The relation \eqref{eq: HL-rec imit} can be obtained directly from the definition: the first term on the r.h.s. corresponds to a monomer on $i$, while the following terms correspond to a dimer on $ij$ for some $j$ neighbour of $i$.

The hard-core interaction is not sufficient to cause a phase transition, but adding also the imitative interaction the system can have phase transitions \cite{Ch,Chcambr,HLliquid,ACM}: in sections \ref{sec: meanfield} we will study this phase transition on the complete graph.
The location of the zeros of the partition function in the complex plane in presence of imitation is an open problem.

\section{Quenched models: Erd\H{o}s-R\'enyi and random field} \label{sec: random}
In this section we consider monomer-dimer models with pure hard-core interactions in some random environment: the randomness is either in the structure of the graph or in the activities.
In the first case we considered a class of random graphs that have locally tree-like structure and finite variance degree distribution \cite{AC}: this is the same for which the ferromagnetic Ising model was rigorously solved by Dembo and Montanari \cite{DM,DMfact}, using the local weak convergence strategy developed in \cite{A}. For the sake of clarity, in this review we have chosen to present the results on the Erd\H{o}s-R\'enyi random graph, but they easily extend for example to random regular graphs and configuration models.

\subsection{Self-averaging for monomer-dimer models} \label{sec: self-av}
One of the most important property describing the effects  of randomness in statistical mechanics models is the \textit{self-averaging} of physical quantities. Under quite general hypothesis a monomer-dimer model with independent random weights has a self-averaging pressure density \cite{ACMrand}.

Let $w_{ij}^{(N)}\geq0\,$, $1\leq i<j\leq N$, $N\in\N$, and $x_i>0\,$, $i\in \N$, be \textit{independent} random variables and consider the (random!) partition function
\eq
Z_N \,=\, \sum_{D\in\DD_N}\, \prod_{ij\in D}w_{ij}^{(N)} \!\!\prod_{\,i\in M_N(D)}\!\!\!\!x_i \;.
\eeq
Since the dimer weights may depend on $N$ and may take the value zero, this framework is very general.
Denote simply by $\E[\,\cdot\,]$ the expectation with respect to all the weights and assume
\eq \label{eq: self-av hyp}
\sup_N\sup_{ij}\E[w_{ij}^{(N)}]=:C_1<\infty\,, \quad \sup_{i}\,\E[x_i]=:C_2<\infty\,, \quad \sup_{i}\,\E[x_i^{-1}]=:C_3<\infty \;.
\eeq
The pressure density $p_N:=\frac{1}{N}\log Z_N$ is a random variable with finite expectation, indeed
\[ N\,p_N \,\begin{cases}
\,\geq\, \log \prod_{i=1}^N x_i \,=\, \sum_{i=1}^N \log x_i \,\geq\, \sum_{i=1}^N (1+x_i^{-1}) \; \in L^1(\Prob)\\
\,\leq\, Z_N-1\; \in L^1(\Prob) \end{cases} \,. \]
The following theorem shows that in the limit $N\to\infty$ the pressure density $p_N$ concentrates around its expectation.

\begin{theorem}[see \cite{ACMrand}] \label{thm: self-av}
For all $t>0$, $N\in\N$, $q\geq1$
\begin{equation} \label{eq: self-av1}
\Prob\big( \,|p_N-\E[p_N]| \geq t \big) \,\leq\,
2\,\exp\bigg(-\frac{t^2\,N}{4\,q^2\,\log^2 N} \bigg) \,+\, (a+b\,N)\,N^{1-q} \;,
\end{equation}
where $a:=4+2C_2C_3\,$, $b:=2C_1C_3^2\,$.
As a consequence, choosing $q>3$,
\begin{equation} \label{eq: self-av2}
|p_N-\E[p_N]| \xrightarrow[N\to\infty]{} 0\,\ \Prob\text{-almost surely} \;.
\end{equation}
\end{theorem}

If the random variables $w_{ij}^{(N)},\,x_i,\,x_i^{-1}$ are bounded, then one can obtain an exponential rate of convergence instead of \eqref{eq: self-av1}.

\subsection{Erd\H{o}s-R\'enyi random graph}

Let $G_N$ be a Erd\H{o}s-R\'enyi random graph over $N$ vertices: each pair of vertices is connected by an edge independently with probability $c/N>0$.
Denote by $Z_N(x)$ the partition function of a monomer-dimer model with monomer activity $x>0$ and pure hard-core interaction on the graph $G_N$:
\eq \label{eq: Z unif w}
Z_N(x) \,=\, \sum_{D\in\DD_{G_N}} x^{N-2|D|} \;,
\eeq
$\langle\,\cdot\,\rangle_{G_N,x}$ will be the corresponding Gibbs expected value. The pressure density is
\eq
p_N(x) \,:=\, \frac{1}{N}\,\log Z_N(x) \;,
\eeq
and the monomer density is
\begin{equation}
m_N(x) \,:=\, \Big\langle\frac{N-2|D|}{N}\,\Big\rangle_{N,x} \,=\, x\,\frac{\partial p_N}{\partial x}(x) \;.
\end{equation}
Since the set of configurations $\DD_{G_N}$ is random, the partition function, the pressure density and the monomer density are random variables and the Gibbs measure is a random measure. This randomness is treated as quenched with respect to the thermal fluctuations.

\begin{theorem}[see \cite{AC,S}]\label{teor: treelike}
Almost surely and for all $x>0$ the monomer density and the pressure density converge in the thermodynamical limit. Precisely:
\eq \label{eq: mNlim}
m_N(x)\ \xrightarrow[N\to\infty]{\textrm{a.s.}}\ \E[M(x)]
\eeq
\eq \label{eq: pNlim}
p_N(x) \ \xrightarrow[N\to\infty]{\textrm{a.s.}}\ \E\big[\log\frac{M(x)}{x}\big] \,-\, \frac{c}{2}\;\E\big[\log\big(1+\frac{M_1(x)}{x}\;\frac{M_2(x)}{x}\big)\big] \;.
\eeq
The law of the random variable $M(x)$ is the only solution supported in $[0,1]$ of the following fixed point distributional equation:
\eq \label{eq: Mfixed}
M\overset{d}{=}\frac{x^2}{x^2+\sum_{i=1}^{\Delta} M_i}
\eeq
where $(M_i)_{i\in\N}$ are i.i.d. copies of $M$ and $\Delta$ is an independent Poisson$(c)$-distributed random variable.
The limit monomer density and the limit pressure density are analytic functions of the activity $x>0$.
\end{theorem}

The expression for the pressure on the right hand side of \eqref{eq: pNlim} was provided by Zdeborov\'a and M\'ezard \cite{ZM} via the theoretical physics method of cavity fields. This theorem provides a complete rigorous proof of their conjecture, partially studied in \cite{BLS,S}.

\begin{figure}[h]
\centering
\includegraphics[scale=0.44]{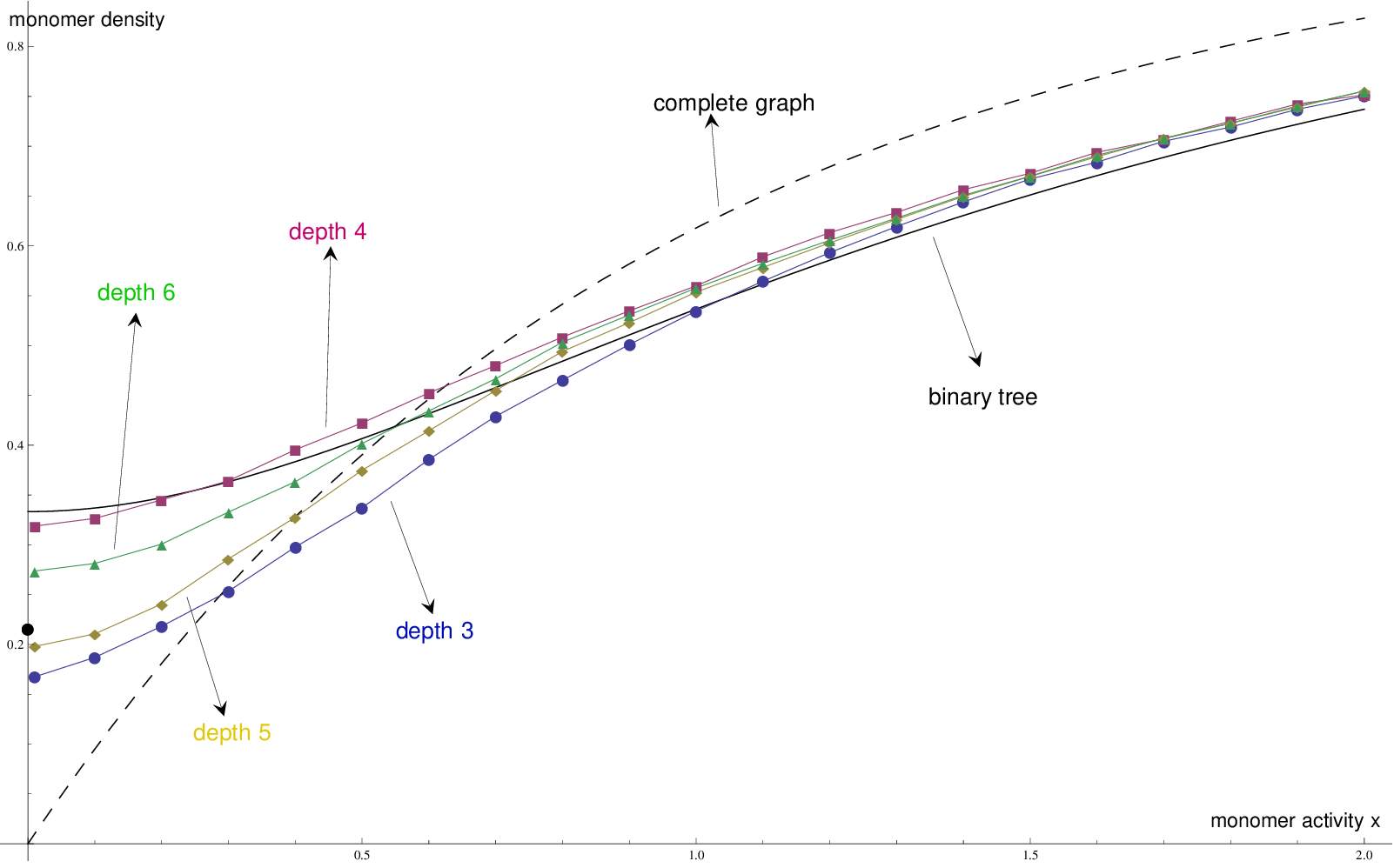}
\caption{Upper (even depths) and lower (odd depths) bounds for the limit monomer density $m(x)=\lim_{N\to\infty}m_N(x)$ versus the monomer activity $x$, in the Erd\H{o}s-R\'enyi case with $c=2$. The binary tree (continuous line) and the complete graph (dashed line) cases are also shown.
The distributional recursion \eqref{eq: Mfixed} is iterated a finite number $r$ of times with initial values $M_i\equiv 1$: the obtained random variable $M_x(r)$ represents the root monomer probability of the random tree $[T,o]_r$. For values of $x=0.01,\,0.1,\,0.2,\,\dots,\,2$, the random variables $M_x(r),\,r=3,\,4,\,5,\,6$ are simulated numerically $10000$ times and an empirical mean is taken in order to approximate $\E[M_x(r)]$. $\E[M_x(r)]$ provides an upper/low approximation of $m(x)$ when $r$ is even/odd.}
\end{figure}

The proof of theorem \ref{teor: treelike} relies on the locally tree-like structure of the Erd\H{o}s-R\'enyi random graphs.
Precisely fix a radius $r\in\N$ and for any vertex $v$ denote by $[G_N,v]_r$ the ball of center $v$ and radius $r$ in the graph $G_N$; then consider a random tree $T$ rooted at the vertex $o$ and with independent Poisson$(c)$-distributed offspring sizes;
it holds (see \cite{DMbraz}):
\eq
\frac{1}{N}\,\sum_{v\in G_N} F\big([G_N,v]_r\big)\ \xrightarrow[n\to\infty]{a.s.}\
\E F\big([T,o]_r\big)
\eeq
for every bounded real function $F$ invariant under rooted graph isomorphisms.

Clearly the monomer density rewrites as an average over the vertices:
\eq \begin{split}
& m_N(x) = \frac{1}{N}\sum_{v\in G_N} \cal M_x(G_N,v) \ ,\quad\text{where} \\
& \cal M_x(G_N,v) := \langle\1(v\textrm{ is a monomer})\rangle_{G_N,x} \;.
\end{split}\eeq
A priori $\cal M_x(G_N,v)$ depends on the whole graph $G_N$, but it can be substituted by local quantities thanks to the following correlation inequalities:
\begin{lemma}[Correlation inequalities]\label{lemma: localisation}
Let $(G,o)$ be a rooted graph, let $r\in\N$. If $[G,o]_{2r+1}$ is a tree, then
\eq
\cal M_x([G,o]_{2r+1}) \,\leq\, \cal M_x(G,o) \,\leq\, \cal M_x([G,o]_{2r}) \;.
\eeq
\end{lemma}
Therefore one can deduce that
\eq \label{eq: mNlim1}
m_N(x) \ \xrightarrow[N\to\infty]{a.s.}\ \lim_{r\to\infty} \E \cal M_x\big([T,o]_r\big)
\eeq
provided the existence of the $\lim_{r\to\infty}$. In this way the problem on random graphs is reduce to the study of the root monomer probability on a random tree. As usual in Statistical Mechanics working on trees is much easier since there are no loops in the interactions.

The problem is now approached by means of the Heilmann-Lieb recursion.
By lemma \ref{lemma: localisation}, the sequences of monomer probabilities respectively at even and odd depths of the tree are monotonic:
\eq
\cal M_x\big([T,o]_{2r}\big) \nearrow M_\textrm{even}(x)\ ,\quad
\cal M_x\big([T,o]_{2r+1}\big) \searrow M_\textrm{odd}(x) \quad\text{ as }r\to\infty \;.
\eeq
The relation \eqref{eq: HL rec} for partition functions gives the following relation for root monomer probabilities:
\eq \label{eq: Mrec}
\left(\begin{array}{c} M_\textrm{even}(x) \\ M_\textrm{odd}(x) \end{array}\right)
\,\overset{d}{=}\,
\left(\begin{array}{c} \frac{x^2}{x^2+\sum_{i=1}^\Delta M_\textrm{odd}^{(i)}(x)} \\ \frac{x^2}{x^2+\sum_{i=1}^\Delta M_\textrm{even}^{(i)}(x)} \end{array}\right)
\eeq
where $(M_\textrm{even}^{(i)},\,M_\textrm{odd}^{(i)})$, $i\in\N$, are i.i.d. copies of $(M_\textrm{even},\,M_\textrm{odd})$.
A direct computation from equation \eqref{eq: Mrec} shows that
\eq
\E[|M_\textrm{even}(x)-M_\textrm{odd}(x)|] \,\leq\, \frac{c^2}{x^4}\, \E[|M_\textrm{even}(x)-M_\textrm{odd}(x)|]
\eeq
therefore $M_\textrm{even}(x)=M_\textrm{odd}(x)$ almost surely for every $x>\sqrt{c}\,$.
Now allow the monomer activity to take complex values in $\HH_+ = \{z\in\C\,|\,\Re(z)>0\}$. This has no physical or probabilistic meaning, but it is a technique to obtain powerful results by exploiting complex analysis. Using the Heilmann-Lieb recursion one can prove that for any rooted graph $(G,o)$, the function $\cal M_z(G,o)$ is analytic in $z\in\HH_+$ and is uniformly bounded as $|\cal M_z(G,o)| \leq |z|/\Re(z)$. It follows that the limit functions $M_\textrm{even}(z)$ and $M_\textrm{odd}(z)$ are analytic on $\HH_+$. Therefore by uniqueness of the analytic continuation
\eq
M_\textrm{even}(x)=M_\textrm{odd}(x)=:M(x) \quad \text{almost surely for every }x>0
\eeq
and \eqref{eq: mNlim} follows by \eqref{eq: mNlim1}.
$M(x)$ satisfies the distributional fixed point equation \eqref{eq: Mfixed}.
The solution supported in $[0,1]$ is unique, since for any random variable $M'\in[0,1]$ that satisfies \eqref{eq: Mfixed} it can be shown by induction on $r\in\N$ that
\eq
\cal M_x\big([T,o]_{2r+1}\big) \,\overset{d}{\leq}\, M' \,\overset{d}{\leq}\, \cal M_x\big([T,o]_{2r}\big) \;.
\eeq

These are the ideas to prove the convergence of the monomer density. To complete the theorem \ref{teor: treelike} it remains to prove the convergence of the pressure density.
The convergence of $p_N(x)$ to some function $p(x)$ is guaranteed by the convergence of its derivative $m_N(x)/x$ together with the bounds \ref{eq0: p bounds}.
Call $\tilde p(x)$ the function defined by the right hand side of \eqref{eq: pNlim}, which can be \virg{guessed} by the heuristic method of energy shifts. Direct computations show that $x\,\tilde p'(x) = m(x) = x\,p'(x)$ for every $x>0$ and $\lim_{x\to\infty} \tilde p(x)-\log x = 0 = \lim_{x\to\infty} p(x)-\log x\,$. Therefore $p=\tilde p$.

\subsection{Random Field}
For the class of models described above the randomness is in the graph structure. The model below instead introduces
a randomness in the monomer activities and is useful to describe impurities. Consider the pure hard-core monomer-dimer model defined in \ref{df0: md model} and assume that $G=(V,E)$ is the complete graph with $N$ vertices, the monomer activities $(x_i)_{i\in V}$ are \textit{i.i.d.} positive random variables and the dimer activity is uniform $w_{ij}=w/N>0\; \forall i\neq j\in V$. The partition function is
\eq \label{eq: Z unif w}
Z_N \,=\, \sum_{D\in\DD_N} \big(\,\frac{w}{N}\,\big)^{|D|} \!\!\prod_{i\in M(D)}\!\!x_i \;.
\eeq
Notice that now the partition function and the pressure density $\frac{1}{N}\log Z_N$ are random variables. The first important fact is  that under the assumptions of Theorem \ref{thm: self-av} the pressure density is self-averaging, namely it converges almost surely to its expectation usually called \textit{quenched pressure density}. The Gaussian representation for the partition function \eqref{eq: gauss repr} and a careful application of the Laplace method allows us to find its limiting value. More precisely  the next theorem shows that thermodynamic limit the quenched pressure density exists and is given by a one-dimensional variational principle, which admits a unique solution.

\begin{theorem}[see \cite{ACMrand}] \label{thm: main}
Let $w>0$. Let $x_i>0,\,i\in\N$ be i.i.d. random variables with $\Ex[x]<\infty$ and $\Ex[(\log x)^2]<\infty$. Then:
\eq \label{eq: var princ}
\exists\ \lim_{N\to\infty} \frac{1}{N}\,\Exv[\,\log Z_N] \;=\; \sup_{\xi\geq0}\Phi(\xi) \,\in\R
\eeq
where
\eq \label{eq: Phi}
\Phi(\xi) \,:=\, -\frac{\xi^2}{2w} + \Ex[\,\log(\xi+x)] \quad\forall\,\xi\geq0 \; ,
\eeq
the function $\Phi$ reaches its maximum at a unique point $\xi^*$ which is the only solution in $\,[0,\infty[\,$ of the fixed point equation
\eq \label{eq: fixed point}
\xi \,=\, \Ex\bigg[\frac{w}{\xi+x}\bigg] \;.
\eeq
\end{theorem}

Theorem \ref{thm: main} allows to compute the main macroscopic quantity of physical interest, that is the \textit{dimer density}, in terms of the positive solution $\xi^*$ of the fixed point equation \eqref{eq: fixed point}.

\begin{corollary} \label{cor: smooth}
In the hypothesis of the theorem \ref{thm: main} the limiting pressure per particle
\eq
 p(w):=\lim_{N\to\infty}\frac{1}{N}\,\Exv\big[\log Z_N(w)\big]
\eeq
exists and is a smooth function of $w>0$. Moreover the limiting dimer density
\eq \label{eq: dimer density}
d:=\lim_{N\to\infty}\frac{1}{N}\,\Exv\big[\big\langle\,|D|\,\big\rangle_{\!N\,}\big] \,=\, w\,\frac{\dd\, p}{\dd w} \,=\, \frac{(\xi^*)^2}{2w} \; .
\eeq
\end{corollary}

A detailed proof of Theorem \ref{thm: main} can be found in \cite{ACMrand}. Here we mention the main ideas.
The first step is to use the Gaussian representation \eqref{eq: gauss repr} for the partition function \eqref{eq: Z unif w} that gives
\eq \label{eq: gauss repr unif w}
Z_N \,=\, \Exi\bigg[\prod_{i=1}^N(\xi+x_i)\bigg] \;,
\eeq
where $\xi$ is a one-dimensional Gaussian random variable with mean $0$ and variance $w/N$. Indeed by proposition \ref{prop: gauss repr}, $Z_N = \E_g \big[\prod_{i=1}^N( g_i+x_i)\big]$ where $g=(g_1,\dots,g_N)$ is an $N$-dimensional Gaussian random vector with mean $0$ and constant covariance matrix $(w/N)_{i,j=1,\dots,N}\,$. It is easy to check that the vector $g$ has the same joint distribution of the  vector $(\xi,\dots,\xi)$ and the identity \eqref{eq: gauss repr unif w} follows. It is important to notice how easily, in this mean-field framework, the Gaussian representation reduces the degrees of freedom of the system. By explicitly rewriting
\eqref{eq: gauss repr unif w} as
\eq  \label{eq: gauss repr unif w 1}
Z_N \,=\, \frac{\sqrt{N}}{\sqrt{2\pi w}}\, \int_\R e^{-\frac{N}{2w}\,\xi^2}\, \prod_{i=1}^N(\xi+x_i) \;\dd\xi  \;.
\eeq
and considering the function
\eq \label{eq: PhiN}
f_N(\xi) \,:=\, e^{-\frac{N}{2w}\,\xi^2}\, \prod_{i=1}^N(\xi+x_i) \quad\forall\,\xi\in\R \;
\eeq
one sees that Theorem \ref{thm: main} follows by approximating $e^\Phi\,$ in the integral \eqref{eq: gauss repr unif w 1}
with the Laplace method.

\section{The mean-field case} \label{sec: meanfield}
Let $h\in\R$ and $J\geq 0$ and consider the  imitative monomer-dimer model in definition \ref{df0: md model-im} within the follwing assumptions: $G=(V,E)$ is the complete graph with $N$ vertices and $\forall i\neq j\in V$  we set $w_{ij}=1/N$,  $x_i\equiv e^h$ and $J_{ij}=J/N$. Since the number of edges is of order $N^2$, in order to keep the logarithm of the partition function of order $N$, a normalisation of the dimer activity as $1/N$ (see Remark \ref{rk0: pressure bounds}) and  the imitation coefficient as $J/N$ are needed.

One can express the Hamiltonian in terms of occupancy variables as
\eq \label{eq: H complete graph}
\Him_N(\alpha) :=\, -h\, \sum_{i=1}^N\,\alpha_i \,-\, \frac{J}{N}\sum_{1\leq i<j\leq N}\big(\alpha_i\,\alpha_j+(1-\alpha_i)\,(1-\alpha_j)\big)
\eeq
for every monomer-dimer configuration on the complete graph $\alpha\in\DD_N$. The partition function is
\eq \label{eq: Z complete graph}
\Zim_N:= \,\sum_{\alpha\in\DD_N} N^{-D_N}\,\exp(-\Him_N) \;,
\eeq
where $D_N:=\sum_{1\leq i<j\leq N}\alpha_{ij}$ represents the total number of dimer for a given configuration $\alpha\in\DD_N$.  Observe that the only relevant quantity in this setting is actually the total number of monomers in a given monomer-dimer configuration
\eq\label{mag}
M_N := \sum_{i=1}^{N}\,\alpha_i
\eeq
indeed the hardcore constraint \eqref{eq0: alphai} implies that $M_N+2D_N=N$ and the Hamiltonian \eqref{eq: H complete graph} is actually a function of $M_N$ only.
We denote the corresponding Gibbs measure as
\eq \label{eq: mu complete graph}
\muim_N(\alpha):= \frac{N^{-D_N(\alpha)}\exp(-\Him_N(\alpha))}{\Zim_N} \quad\forall\,\alpha\in\DD_N
\eeq
and the expectation with respect to the measure $\muim_N$ is denoted by $\langle\,\cdot\,\rangle_N$.
In particular, setting $\mim_N:=\frac{1}{N}\sum_{i=1}^N\alpha_i$, the average monomer density is
\eq \label{eq: m complete graph}
\langle m_N\rangle_N = \,\sum_{\alpha\in\DD_N}\frac{\sum_{i=1}^N\alpha_i}{N}\; \frac{\exp(-\Him_N(\alpha))}{\Zim_N} \,=\,  \frac{\partial}{\partial h}\frac{\log\Zim_N}{N} \;.
\eeq

This model has been initially studied in \cite{ACM,ACMepl}, where the behaviour of the pressure and monomer densities in the thermodynamic limit is analysed.

\begin{theorem}[see \cite{ACM}] \label{thm: imitative limit}
Let $h\in\R,\,J\geq0$. Then there exists
\begin{equation} \label{eq: p limit}
\pim:=\, \lim_{N\to\infty}\frac{\log\Zim_N}{N} \,=\, \sup_{m}\,\psi(m)
\end{equation}
the $\sup$ can be taken indifferently over $m\in[0,1]$ or $m\in\R$ and
\begin{equation} \label{eq: varp}
\psi(m,h,J) :=\, -J m^2 + \frac{J}{2} + \poo(2J m+h-J)
\end{equation}
where for all $t\in\R$
\begin{equation} \label{eq: p0}
\poo(t) :=\, -\frac{1}{2} \big(1-g(t)\big) -\frac{1}{2}\log\big(1-g(t)\big) \;,
\end{equation}
\begin{equation} \label{eq: g}
g(t) \,:=\, \frac{1}{2}\,(\sqrt{e^{4t}+4\,e^{2t}}-e^{2t})  \;.
\end{equation}
Furthermore the function $\psi(m)$ attains its maximum in (at least) one point $m^*=m^*(h,J)\in(0,1)$, which is a solution of the the consistency equation
\begin{equation} \label{eq: consistency equation}
m \,=\, g\big((2m-1)\,J+h\big) \;.
\end{equation}
At each value of the parameters $(h,J)$ such that $h\mapsto m^*(h,J)$ is differentiable, the monomer density admits thermodynamic limit
\begin{equation} \label{eq: m limit}
\lim_{N\to\infty} \langle m_N\rangle_N \,=\, m^* \;.
\end{equation}
\end{theorem}

In order to prove Theorem \ref{thm: imitative limit}, first we need to deal with the case $J=0$, then the limit \eqref{eq: p limit} with $J>0$ follows by a convexity argument introduced by Guerra \cite{Gue} for the Curie-Weiss model.
At $J=0$ the model reduces to the pure monomer-dimer model of definition \ref{df0: md model} on the complete graph with  $x_i=x=e^h>0$, $w_{ij}=1/N>0\;\forall i\neq j\in V$. Let us denote by $\Zcl_N$ and $\langle m_N\rangle_N^{(0)}$ respectively the partition function and the average monomer density at $J=0$; it holds
\eq \label{eq: convergence0}
\lim_{N\to\infty} \frac{1}{N}\,\log \Zcl_N \,=\, \poo(h)
\eeq
\begin{equation}
\lim_{N\to\infty}\langle m_N\rangle_N^{(0)} \,=\, g(h) \;.
\end{equation}
The function $\poo$ is analytic thus at $J=0$ there are no phase transition in agreement with the general result of Heilman-Lieb \cite{HL}.
The limit \eqref{eq: convergence0} can be obtained in two different ways:
\begin{itemize}
\item[1)] by a combinatorial computation, since on the complete graph it is possible to compute explicitly the number of monomer-dimer configurations with a given number of monomers;
\item[2)] by using the Gaussian representation \eqref{eq: gauss repr} of the partition function and the Laplace method.
\end{itemize}
The latter method furnish a better estimation of the convergence \eqref{eq: convergence0}
\eq \label{eq: convergence0 Lap}
\Zcl_N(h) \,\underset{N\to\infty}{\sim}\, \frac{\exp\big(N\poo(h)\big)}{\sqrt{2-g(h)}} \;,
\eeq
which will be fundamental in the study of the fluctuations of $M_N$ (Section 5).

\begin{remark}
The limiting pressure density $p$ can also be expressed as a different variational problem, equivalent to that of Theorem \ref{thm: imitative limit}:
\begin{equation} \label{eq: varp limit}
p = \sup_m \big( s(m)-\varepsilon(m) \big)
\end{equation}
with
\begin{gather}\label{entropyenergy}
s(m)\,:=\, -m\,\log m -\frac{1-m}{2}\,\log (1-m) + \frac{1+m}{2}\\
\varepsilon(m):= -J\, m^2- (h-J)\,m -\frac{J}{2} \;.
\end{gather}
The variational problem \eqref{eq: varp limit} can be obtained directly by the combinatorial computation mentioned before.
The function $s$ and $\varepsilon$ in \eqref{entropyenergy} are the \textit{entropy and energy densities} respectively.
\end{remark}

The properties of the solution(s) of the one-dimensional variational problem \eqref{eq: p limit} appearing
in theorem \ref{thm: imitative limit} determine the thermodynamic properties of the model.
In particular we are interested in the value(s) of $m=m^*(h,J)$ where the maximum is reached, since it can be interpreted as the limiting value of the monomer density.

The function $m^*$ (see \cite{ACM}) is single-valued and smooth on the plane $(h,J)$ with the exception of an implicitly defined open curve $h=\gamma(J)$ and its endpoint $(h_c,J_c)$. Instead on $\gamma$ there are two global maximum points $m_1<m_2$ that correspond to the \textit{dimer} phase and the \textit{monomer} phase respectively. Crossing the curve $\gamma$ in the phase plane the order parameter $m^*$ presents a jump discontinuity: in other words a \textit{second order phase transition} occurs and $\gamma$ is the coexistence curve.
The point $(h_c,J_c)$ is the \textit{critical point} of the system, where $m^*$ is continuous but not differentiable.

\begin{figure}[h]
\centering
\includegraphics[scale=1]{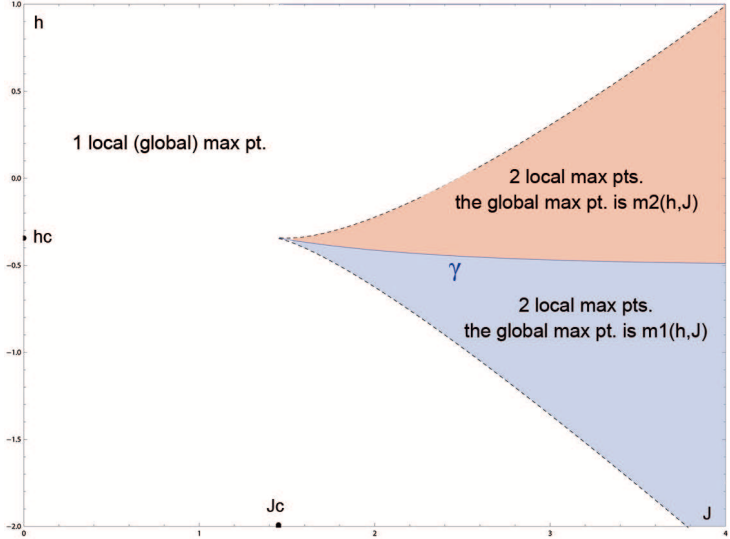}
\caption{Phase space $(h,J)$. The curve $\gamma$ separates the values $(h,J)$ for which the global maximum point $m^*(h,J)$ of $m\mapsto\vpim(m,h,J)$ jumps between two values $m_1<m_2$. This entails a discontinuity of $m^*(h,J)$ along the coexistence curve $\gamma$.}
\end{figure}

\begin{remark}
We notice that the techniques developed in \cite{ACM} do not allow us to conclude the existence of the limiting
monomer density on the coexistence curve $\gamma$.  In the standard mean-field
ferromagnetic model (Curie-Weiss) the existence of the magnetization on the coexistence curve ($h=0$)
follows directly by the global spin flip symmetry, a property that we do not have in the present case.
\end{remark}

The non analytic behaviour of $m^*(h,J)$ near the critical point is described by its \textit{critical exponents}.

\begin{theorem}[see \cite{ACM}] \label{thm: critical exponents}
Consider the global maximum point $m^*(h,J)$ of the function $m\mapsto\vpim(m,h,J)$ defined by (\ref{eq: varp limit}). Set $m_c:=m^*(h_c,J_c)$. The critical exponents of $m^*$ at the critical point $(h_c,J_c)$ are:
\[ \boldsymbol{\beta}\, = \lim_{J\to J_c+}\frac{\log|m^*(\delta(J),J)-m_c|}{\log(J-J_c)}\; =\, \frac{1}{2} \]
along any curve $h=\delta(J)$ such that $\delta\in C^2([J_c,\infty[)$, $\delta(J_c)=h_c$, $\delta'(J_c)=\gamma'(J_c)$ (i.e. if the curve $\delta$ has the same tangent of the coexistence curve $\gamma$ at the critical point);\\
\begin{gather*}
\frac{1}{\boldsymbol{\delta}}\, = \lim_{J\to J_c}\frac{\log|m^*(\delta(J),J)-m_c|}{\log|J-J_c|}\; =\, \frac{1}{3} \\[2pt]
\frac{1}{\boldsymbol{\delta}}\, = \lim_{h\to h_c}\frac{\log|m^*(h,\delta(h))-m_c|}{\log|h-h_c|}\; =\, \frac{1}{3}
\end{gather*}
along any curve $h=\delta(J)$ such that $\delta\in C^2(\R_+)$, $\delta(J_c)=h_c$, $\delta'(J_c)\neq\gamma'(J_c)$ or along a curve $J=\delta(h)$ such that $\delta\in C^2(\R)$, $\delta(h_c)=J_c$, $\delta'(h_c)=0$ (i.e. if the curve is not tangent to $\gamma$ at the critical point).
\end{theorem}
Theorem \ref{thm: critical exponents} proves that the model belongs to the same \textit{universality class} of the mean-field
ferromagnet.

\section{Distributional limit theorems at the critical point}\label{sec: chuck}
In this section we study the distributional limit of the random variable \textit{number of monomers} with respect to the Gibbs measure on the complete graph \cite{ACFM,ACFMepl}. We show that a \textit{law of large numbers} holds outside the coexistence curve $\gamma$, whereas on $\gamma$ the limiting distribution is a convex combination of two Dirac deltas representing the two phases (theorems \ref{teo1}, \ref{teorema.1}). Moreover we show that a \textit{central limit theorem} holds outside $\gamma\cup(h_c,J_c)\,$, while at the critical point a normalisation of order $N^{-3/4}$ for the fluctuations is required and the limiting distribution is $Ce^{-cx^4}\dd x$ (theorems \ref{teo1}, \ref{teorema.2}).

In \cite{ACFM} we follow the Gaussian convolution method introduced by Ellis and Newman for the mean-field Ising model (Curie-Weiss) in \cite{ellis1978limit,ellis1978statistics,ellis1980limit} in order to deal with the imitative potential.
An additional difficulty stems from the fact that even in the absence of imitation the system keeps an interacting nature due to the presence  of the hard-core interaction: we use the Gaussian representation \ref{prop: gauss repr} to \textit{decouple} the hard-core interaction.

We focus on the behaviour of the random variable $M_N=\sum_{i=1}^N\alpha_i$ (number of monomers) with respect to the Gibbs measure \eqref{eq: mu complete graph} with a suitable scaling when $N\to\infty\,$.
From now on $\delta_x$ is the Dirac measure centred at $x$, $\mathcal{N}\left(x,\sigma^2\right)$ denotes the Gaussian distribution with mean $x$ and variance $\sigma^2$ and $\overset{\mathcal{D}}{\rightarrow}$ denotes the convergence in distribution with respect to the Gibbs measure $\mu_N$ as $N\to\infty\,$.


At $J=0$ the law of large numbers and the central limit theorem hold true:

\begin{theorem}[see \cite{ACFM}]\label{teo1}
At $J=0$ the following results hold:
\begin{equation}\label{LLN0}
\frac{M_N}{N} \,\overset{\mathcal{D}}{\rightarrow}\, \delta_{g(h)}
\end{equation}
and
\begin{equation}\label{CLT0}
\frac{M_{N}-N\,g(h)}{\sqrt{N}} \,\overset{\mathcal{D}}{\rightarrow}\, \mathcal{N}\left(0,\,\frac{\partial g}{\partial h}(h)\right)
\end{equation}
where $g$ is the function defined by \eqref{eq: g}.
\end{theorem}

Notice that, even if $J=0$, \eqref{CLT0} is not a consequence of the standard central limit theorem, indeed $M_N$ is not a sum of i.i.d. random variables because of the presence of the hard-core interaction.
The theorem \ref{teo1} follows from the recent results of Lebowitz-Pittel-Ruelle-Speer \cite{Lebowitz}. A different proof is presented here which includes also the general value of $J>0$. We should mention that a slightly improvement of the result presented has been obtained with different methods in \cite{We}.

Consider the asymptotic behaviour of the distribution of the number of monomers $M_N$ with respect to the Gibbs measure $\mu_N$.
The \textit{law of large numbers} holds outside the coexistence curve $\gamma$, on $\gamma$ instead it breaks down in a convex combination of two Dirac deltas. Precisely it holds
\begin{theorem}[see \cite{ACFM}]\label{teorema.1}
\begin{itemize}
\item[i)] In the uniqueness region $(h,J)\in\big(\R\times\R_+\big)\setminus\gamma$, denoting by $m^*$ the unique global maximum point of the function $\tilde p(m)$ defined by \eqref{eq: varp}, it holds
\begin{equation}\label{pequal1}
	\frac{M_N}{N}\overset{\mathcal{D}}{\rightarrow}\delta_{m^*}
\end{equation}
\item[ii)] On the coexistence curve $(h,J)\in\gamma$, denoting by $m_1,m_2$ the two global maximum points of $\tilde p(m)$, it holds
\begin{equation}\label{pequal2}
	\frac{M_N}{N}\overset{\mathcal{D}}{\rightarrow} \rho_1\,\delta_{m_{1}}+ \rho_2\,\delta_{m_{2}} \;,
\end{equation}
where $\rho_l=\frac{b_{l}}{b_{1}+b_2}\,$, $b_{l}=(-\lambda_l(2-m_l))^{-1/2}\,$  and
$\lambda_l=\frac{\partial^2\widetilde{p}}{\partial m^2}(m_l)\,$, for $l=1,2$.
\end{itemize}
\end{theorem}

\begin{remark}\label{www}
We notice that, on the contrary of what happens for the Curie-Weiss model, the statistical weights $\rho_{1}$ and $\rho_{2}$  on the coexistence curve  are in general different, furthermore they are not simply given in terms of the second derivative of the variational pressure $\widetilde p\,$.

The first fact can be seen  numerically, 
and analytically one can compute
\begin{equation}\label{boundgamma}
\lim_{J\to \infty}\frac{\rho_1(J)}{\rho_2(J)}=\frac{1}{\sqrt 2} \;.
\end{equation}

The second fact can be interpreted as follows: the relative weights $\rho_l$ have two contributions reflecting the presence of two different kinds of interaction. The first contribution $\lambda_l$ is given by the second derivative of the variational pressure \eqref{eq: varp}, while the second contribution $2-m_l$ comes from the second derivative of the pressure of the pure hard-core model.
\end{remark}


The central limit theorem holds outside the union of the coexistence curve $\gamma$ and the critical point $(h_c,J_c)$. At the critical point its breakdown results in a different scaling $N^{3/4}$ and in a different limiting distribution $Ce^{-cx^4}\dd x\,$. Precisely

\begin{theorem}[see \cite{ACFM}]\label{teorema.2}
\begin{itemize}
\item[i)] Outside the coexistence curve and the critical point $(h,J)\in\big(\R\times\R^+\big)\setminus\big(\gamma\cup(h_c,J_c)\big)$, it holds
\begin{equation}\label{CL1}
\frac{M_{N}-N m^*}{N^{1/2}} \,\overset{\mathcal{D}}{\rightarrow}\, \mathcal{N}\Big(0,\sigma^2\Big)
\end{equation}
where  $\sigma^2=-\lambda^{-1}-(2J)^{-1}>0\,$ and $\lambda=\frac{\partial^2\widetilde{p}}{\partial m^2}(m^*)<0\,$.
\item[ii)] At the critical point $(h_c,J_c)$, it holds
\begin{equation}\label{CLc}
	\frac{M_{N}-N m_c}{N^{3/4}} \,\overset{\mathcal{D}}{\rightarrow}\, C\,\exp\bigg(\dfrac{\lambda_c}{24}\,x^{4}\bigg)\dd x
	\end{equation}
where $\lambda_c=\frac{\partial^4\widetilde{p}}{\partial m^4}(m_c)<0$, $m_c\equiv m^*(h_c,J_c)$
and $C^{-1}=\int_{\R}\exp (\frac{\lambda_c}{24}x^{4})\dd x\,$.
\end{itemize}
\end{theorem}

The first step to obtain these results is to perform a Gaussian convolution, following the ideas of Ellis and Newmann \cite{ellis1978limit, ellis1978statistics}, in order to decouple the imitative interaction.
Precisely taking $W \sim \mathcal{N}(0,(2J)^{-1})$ a random variable independent of $M_{N}$ for all $N\in\N$, for all $\eta\geq0$ and $u\in\R$, a direct computation shows that the distribution of
\begin{equation}
\frac{W}{N^{1/2 -\eta}}+\frac{M_{N}-N u}{N^{1-\eta}}
\end{equation}
is
\begin{equation}\label{tesi.lemma.2}
C_N\,\exp\Big(N\,\tilde p_N\Big(\dfrac{x}{N^{\eta}}+u\Big)\Big)\,dx \,,
\end{equation}
where $C_N^{-1}=\int_{\mathbb{R}}\exp\big(N\,\tilde p_N(\frac{x}{N^{\eta}}+u)\big)dx\,$,
\begin{equation}\label{funzione.F}
\tilde p_N(x):=-Jx^2+\frac{J}{2}+\pclN(2Jx+h-J)
\end{equation}
and $\pclN(t)=\frac{1}{N}\log\ZclN(t)$ denotes the pressure density of the monomer-dimer system at imitation potential $J=0$ and monomer field $t$.
Therefore we are interested in the limit as $N\to\infty$ of quantities like
\eq \label{eq: intphi}
\int_\R \exp\Big(N\,\tilde p_N\Big(\dfrac{x}{N^{\eta}}+u\Big)\Big)\ \phi(x)\,dx \quad,\quad\phi\text{ bounded continuous}
\eeq
which depends crucially on the scaling properties of $\tilde p_N$ near its global maximum point(s).
Thanks to the Gaussian representation at $J=0$, and precisely from eq.\eqref{eq: convergence0 Lap}, we know that $\tilde p_N$ converges to $\tilde p$ in a very strong way, which allows to replace the Taylor expansion of $\tilde p_N$ by that of $\tilde p$.

\section{Conclusions and outlooks}
The relation of the class of models presented so far with the physically relevant ones in finite dimensional
lattices represent an interesting research problem that can be carried out following the steps of the studies
done for the ferromagnetic spin models \cite{Kac, Tho}. We want to point out, moreover, that the range
of direct applications of mean-field models like these ones is quite developed and quickly
expanding. To make a few examples: the diluted mean-field case studied in Section 3 is
directly related to the matching problem studied in computer sciences \cite{KS}. The model
with attractive interaction studied in Section 5 has been applied to the social sciences \cite{BCSV}.
There is also a growing set of applications of monomer-dimer models to the study of
socio-technical data from novel communication systems like voip conference calls and
messaging \cite{ACMM}. At each single time every user cannot be in more than a call, i.e. the occupation
number fullfills a hard-core constraint. While the old style phone calls were well described
by a monomer-dimer system the novel technological devices needs a wider space of higher
dimensional polymers that allow the presence of multiple individuals in the same virtual room:
the monomer correspond to a silent user, the dimer is a two-body conversation, the trimer a three-body
and so on. The models to be investigated in this case are therefore polymer models with hard-core
interaction on hypergraphs with no physical dimension, i.e. better described as some form of dilution
of the complete hypergraph. The mean-field case and its diluted versions are therefore at the heart of
the problem and not only mere approximations.

\end{document}